\documentclass[10pt,a4paper,reqno]{amsart}
\usepackage{amsthm}
\usepackage{amsmath,mathtools}
\usepackage{amssymb}
\usepackage{graphicx,color}

\usepackage[font=small]{caption}
\usepackage{subcaption}

\usepackage{multirow}
\usepackage[shortlabels]{enumitem}

\newcounter{smalllist}
\newenvironment{SL}{\begin{list}{{\rm\alph{smalllist})}}{%
\setlength{\topsep}{0mm}\setlength{\parsep}{0mm}\setlength{\itemsep}{0mm}%
\setlength{\labelwidth}{2em}\setlength{\leftmargin}{2em}\usecounter{smalllist}%
}}{\end{list}}

\makeatletter
\theoremstyle{plain}
\newtheorem{thm}{Theorem}
  \theoremstyle{definition}
  \newtheorem*{thm*}{Theorem}
  \newtheorem{defn}[thm]{Definition}
  \theoremstyle{remark}
  \newtheorem{rem}[thm]{Remark}
  \theoremstyle{plain}
  
  \theoremstyle{plain}
  \newtheorem{lem}[thm]{Lemma}
  \theoremstyle{plain}
  
 \theoremstyle{definition}
  
  \theoremstyle{remark}
  \newtheorem*{rem*}{Remark}

  \theoremstyle{definition}
\newtheorem{ass}{Assumption}
  \theoremstyle{definition}

  \theoremstyle{definition}
\newtheorem*{nota*}{Notation}
  \theoremstyle{definition}

\usepackage{amsfonts}
\usepackage{mathrsfs}

\newtheorem*{question*}{\it{QUESTION}}

\theoremstyle{plain}

\newcommand{\N}{\mathbb{N}}
\newcommand{\R}{{\mathbb{R}}}

\makeatother

\begin{document}

\title[]{Eigenstates with Infinite Position Moments}

\author{Michal Jex}
\address[Michal Jex]{
	Department of Physics, Faculty of Nuclear Sciences and Physical Engineering, Czech Technical University in Prague,  Břehová 7,115 19 Praha 1, Czech Republic
	}	
\email{michal.jex@fjfi.cvut.cz}



\date{\today}

\begin{abstract}
We prove necessary and sufficient conditions for the Schrödinger operators to have zero-energy bound states at the threshold of the essential spectrum such that they have bounded $k$-th moment. This result is the extension of the results published in D.~Hundertmark, M.~Jex, and M.~Lange  [\emph{Forum Mathematics, Sigma} \textbf{11} (2023)].
\end{abstract}

\maketitle

\section{Introduction}
One of the important aspects of a quantum system is the existence of bound states and their properties related to the stability and localisation of the system. We consider a system described by a Schrödinger operator
\begin{align}
\label{eq:operator}
H=-\Delta+V
\end{align}
on $L^2(\mathbb R^d)$ where $V\in L^1_{\mathrm{loc}}(\R^d)$ is a real-valued potential. The operator can be defined via KLMN theorem as a unique self-adjoint operator related to a closed symmetric quadratic form. The precise conditions on the potential $V$ are given in Assumption \ref{potential} below. The conditions are chosen in such a way that the eigenstates of our operator are continuous.

We are interested in the special case when the ground state is at the threshold of the essential spectrum and has finite $c$-th moment. The existence of such bound states were answered in \cite{HunJexLan22-Potential}. In the present paper we provide sharp sufficient and necessary conditions for the states to have finite $c$-th moment. This is trivial for the eigenfunctions associated to the discrete eigenvalues but non-trivial to the embedded eigenvalues. The special interest has the first moment, i.e., $\langle\psi,|x|\psi\rangle$. The finiteness of the first moment guarantees that the particle ``lives" in the finite region almost certainly, in the sense of its average distance to the origin. The finiteness of the first moment is slightly stronger condition compared to the square integrability of the eigenfunction which corresponds to the existence of associated eigenfunction. Another important quantity is the second moment related to the variance. The method of the proof presented in this paper is applicable for arbitrary moment.

We assume that the essential spectrum has the form $\sigma_{\mathrm{ess}}(H)=\R^+$. This is often the case, especially in physically relevant situation when $\lim_{|x|\rightarrow\infty}V=0$. This means that the eigenvalue which we are interested in lies at the boundary between the discrete and essential spectrum. Point eigenvalues embedded in the continuous spectrum are present only in special specific cases such as slowly decaying oscillating potentials \cite{VonNeuWig-paper-1929}. The conditions for the absence of positive eigenvalues are presented in \cite{Agm70,Sim69}. The existence and behaviour of eigenvalues below the threshold of the essential spectrum is well studied and understood \cite{ReeSim4}. They exhibit exponential localization and they are stabilized by an energy gap.

The study of zero-energy eigenvalues at the threshold of the essential spectrum is quite complicated due to the fact that they are not stabilized by an energy gap, i.e., there is no safety distance to the essential spectrum. Early results on existence and non-existence go back to \cite{JeKa79,Jen80,Jen84,Ram87,Ram88}. It was established that long-range Coulomb repulsion has a stabilisation effect \cite{HofOstHofOstSim83,bol85} which allows the threshold states to exist. This was further extended to potentials decaying slower  than $\frac3{4|x|^2}$ in dimension 3 \cite{BenYar90,GriGar07}. Inverse square scaling is an expected behaviour because it corresponds to the scaling of the Laplacian. The recent result presented in \cite{HunJexLan22-Potential} shows that the leading terms for critical potentials behave as
$$
\frac{d(4-d)}{4|x|^2}+\frac{1}{|x|^2\ln|x|}
$$
where $d$ denotes the dimension. This is an interesting result because there is a phase transition in dimension 4, which also explains the dramatic change in the behaviour of the Laplacian for dimension smaller than 4 and larger than 4. The paper \cite{HunJexLan22-Potential} also provides a much simpler proof compared to the previous results which heavily relied on careful resolvent estimates. The approach from \cite{HunJexLan22-Potential} is also applicable for the case of discrete Schrödinger operators \cite{jex25}. Due to this fact we expect that the results presented in the present paper should be applicable also for the discrete case.

Throughout the paper we use the following notation for iterated logarithms and exponentials which simplifies formulas and improves readability.

\begin{nota*}\label{log}
Let $e_0=0$ and $e_{n+1}=e^{e_n}$. We define iterated logarithm $\ln_n(r)$ for $n\in\N$ and $r>e_n$ as $\ln_1(r):=\ln(r)$ and $\ln_{n+1}(r):=\ln(\ln_n(r))$. Note that iterated logarithm defined in this way is positive.
\end{nota*}

As we have already mentioned we define our operator via a quadratic form
$$
\langle\psi,H\psi\rangle:=\langle\nabla\psi,\nabla\psi\rangle+\langle\psi,V\psi\rangle\,.
$$
The second term is to be understood as $\langle\sqrt{|V|}\psi,\mathrm{sgn} V\sqrt{|V|}\psi\rangle$ where $\mathrm{sgn}$ denotes the sign function. To avoid any problems with domain and to guarantee the continuity of the eigenstates we consider only potentials satisfying the following condition.

\begin{ass}\label{potential}
The potential $V$ is in a local Kato-class $K_{d,\mathrm{loc}}(\R^d)$ and its negative part $V_-=\mathrm{sup}(-V,0)$ is form-bounded with relative bound $a<1$ w.r.t. $-\Delta+V_+$, i.e., there exists $0\leq a<1$ and $b\geq0$ such that
\begin{align*}
\|\sqrt{V_-}\psi\|^2\leq a(\|\nabla\psi\|^2+\|\sqrt{V_+}\psi\|^2)+\|\psi\|^2
\end{align*}
for all $\psi\in H^1(\R^d)\cap\mathcal D(\sqrt{V_+})$ where $D(\sqrt{V_+})$ is a domain of the multiplication operator  $\sqrt{V_+}$ on $L^2(\R^d)$.
\end{ass}

Now we are ready to state two main results of this paper in the following theorems. Together they provide sharp necessary and sufficient conditions for the existence of the ground state of a quantum system at the threshold of the essential spectrum with finite $c$-th moment.

\begin{thm}\label{absence}
Let the potential $V$ satisfy the Assumption \ref{potential}, $\sigma(H)=\R^+$ and $c\geq0$. If there exists $m\in\N_0$ and some $R>e_m$ such that
\begin{align}
V(x)\leq\frac{d(4-d)+4c+c^2}{4|x|^2}+\frac{c+2}{2|x|^2}\sum_{j=1}^m\prod_{k=1}^j\left(\ln_k(|x|)\right)^{-1}
\end{align}
for all $|x|>R$ then $0$ is not a ground state eigenvalue of the Schrödinger operator $H$ with an eigenfunction $\psi$ such that $\langle\psi,|x|^c\psi \rangle<\infty$.
\end{thm}

Before stating the complementary result we need one more definition.

\begin{defn}\label{def:ctitical} Let $W\ge 0$ be infinitesimally form bounded with respect to $-\Delta+V_+$. 
Then the potential $V$ is \emph{critical} if the Schr\"odinger operator $H$ 
has spectrum $\sigma(H)=\sigma_{ess}(H)=[0,\infty)$ and for all $W$ the operator $H_W= H-W$  
has essential spectrum  $\sigma_{ess}(H_W)= [0,\infty)$ and 
a negative energy bound state. The potential $V$ is \emph{subcritical} if the Schr\"odinger operator $H$ 
has spectrum $\sigma(H)=\sigma_{ess}(H)=[0,\infty)$ and there exists a non-trivial $W$ such that
the operator $H_W= H-W$  
has only essential spectrum $\sigma_{ess}(H_W)= [0,\infty)$. 
\end{defn}

\begin{rem}
The critical potential is such that any attractive perturbation will create a discrete negative eigenvalue. For subcritical potentials there exists an attractive perturbation which would not create a bound state.
\end{rem}

\begin{thm}\label{existence}
Let the potential $V$ be critical satisfying the Assumption \ref{potential}, $\sigma(H)=\R^+$ and $c\geq0$. If there exists some $m\in\N_0$, $\varepsilon>0$, and $R>e_m$ such that
\begin{align}\label{eq:thm_existence}
V(x)\geq\frac{d(4-d)+4c+c^2}{4|x|^2}+\frac{c+2}{2|x|^2}\sum_{j=1}^m\prod_{k=1}^j\left(\ln_k(|x|)\right)^{-1}+\frac{\varepsilon(c+2)}{2|x|^2}\prod_{k=1}^m\left(\ln_k(|x|)\right)^{-1}
\end{align}
for all $|x|>R$ then $0$ is a ground state eigenvalue of the Schrödinger operator $H$ with an eigenfunction $\psi$ such that $\langle\psi,|x|^c\psi \rangle<\infty$.
\end{thm}

\begin{rem}
We note that the case $c=0$ in Theorems \ref{absence} and \ref{existence} was already proved in \cite[Theorems 1.3 and 1.7]{HunJexLan22-Potential}.
\end{rem}

\begin{rem}
Due to the fact that the leading order term in both Theorems \ref{absence} and \ref{existence} are of the form $\frac{\mathrm{const}}{|x|^2}$,  one can see that the only potential which can have only some moments finite corresponds to the potential behaving as $\frac{\mathrm{const}}{|x|^2}$. Any other potential with leading term in the form $\frac{\mathrm{const}}{|x|^\beta}, \beta\neq2$ will have either all the moments finite or infinite. 
\end{rem}

\begin{rem}
The result presented in Theorem \ref{existence} can be generalized to excited states because its proof does not require strict positivity of the eigenstate unlike the proof for Theorem \ref{absence}. This would correspond to the situation when due to a perturbation the originally discrete eigenvalue would touch the essential spectrum. In such a case the associated eigenstate would seize to be exponentially decaying, but it would still remain quadratically integrable which in fact requires the asymptotic of the potential to satisfy \eqref{eq:thm_existence}.
\end{rem}

\textit{Organisation of the paper} The paper is split into three sections. The following two sections are dedicated to proofs of Theorems \ref{absence} and \ref{existence}. The last section presents a non-trivial example of the family of potentials which posses the zero energy ground states.

\section{Proof of Absence Result (Theorem \ref{absence})}\label{sec:abs}
The proof of the claim relies on the subharmonic comparison principle. The formulation which we use is due to Agmon \cite[Theorem 2.7]{Agm85} which is written below as Theorem \ref{thm:comp-agm} for the convenience of the reader. This specific form has the advantage of minimal regularity assumptions due to its quadratic form formulations. Before we are able to state it, we need a few definitions. In the introduction we have defined the quadratic form domain of the Schrödinger operator $\Delta+V$ as $\mathcal Q(H)=H^1(\mathbb R^d)\cap D(\sqrt{V_+})$. It is straightforward to also introduce the local domain of a quadratic form for functions supported in the region $U$ as 
$$
\mathcal Q_{loc}^U(H)\coloneqq \{\psi\in L_{loc}^2(U):\chi\psi\in\mathcal Q(H)\,\mathrm{for}\,\mathrm{all}\,\chi\in C_0^\infty(U)\}\,.
$$
This allows us to define local weak subsolutions and supersolutions in the following way.

\begin{defn}\label{def:eigenfunctions etc}
Consider Schrödinger operator $H$ defined in \eqref{eq:operator}. Then
\begin{enumerate} [a)]
\item $u$ is a (weak) eigenfunction of the Schr\"odinger operator 
		$H$ with energy $E$ if  $u\in \mathcal Q(H)$ and  
  \begin{equation}\label{eq:weak eigenfunction}
    \langle \varphi, (H-E) u\rangle = 0 	
  \end{equation}
  for every $\varphi\in  C_0^\infty(\R^d)$.
\item $u$ is a (weak) local eigenfunction of the Schr\"odinger operator 
		$H$ with energy $E$ in $U\subset\R^d$  
 	 if $u\in \mathcal Q^U_{loc}(H)$ and  
  \begin{equation}\label{eq:weak local eigenfunction}
    \langle \varphi, (H-E) u\rangle = 0 	
  \end{equation}
  for every $\varphi\in  C_0^\infty(U)$.
\item  $u$ is a supersolution of the Schr\"odinger operator 
		$H$ with energy $E$ in $U\subset\R^d$ if 
  $u\in \mathcal Q^U_{loc}(H)$ and 
  \begin{equation}\label{eq:supersolution}
  \langle \varphi, (H-E) u\rangle \ge 0 	
  \end{equation}
  for every non-negative $\varphi\in C_0^\infty(U)$. 
\item $u$ is a subsolution of the Schr\"odinger operator 
		$H$ with energy $E$ in $U\subset\R^d$  if 
  $u\in \mathcal Q^U_{loc}(H)$ and 
  \begin{equation}\label{eq:subsolution}
  \langle \varphi, (H-E)  u\rangle \le  0 	
  \end{equation}
  for every non-negative $\varphi\in  C_0^\infty(U)$. 
\end{enumerate}
\end{defn}

\begin{thm}[Agmon's version of the comparison principle]\label{thm:comp-agm}
Let $w$ be a positive supersolution of the Schr\"odinger operator $H$ at energy $E$ in a neighborhood of infinity 
$U_R\coloneqq\{x \in \R^d\,:\, |x| > R\}$. Let $v$ be a subsolution of 
$H$ at energy $E$ in $U_R$. Suppose that
\begin{equation}\label{eq:liminf cond}
\liminf_{N\rightarrow\infty}\left(\frac{1}{N^2}\int_{N\leq|x|\leq\alpha N}|v|^2\mathrm d x\right)=0
\end{equation}
for some $\alpha>1$. If for some $\delta>0$ and $0\le C<\infty$ one has 
\begin{equation}\label{eq:a-priori comparison on annulus}
v(x)\leq Cw(x) \,  \text{ on the annulus } R<|x|\le R+\delta \, ,
\end{equation}
 then 
\begin{equation}
v(x)\leq Cw(x) \, \text{ for all } x\in U_{R}\,.
\end{equation}
\end{thm}

\begin{rem}\label{rem:l2}
Note that the condition \eqref{eq:liminf cond} is trivially satisfied for $L^2$ functions. However in the course of the proof we will also need to use functions which are not necessary square integrable.
\end{rem}

For $m\in\N_0$ and $c\geq0$ we define an auxiliary function
\begin{equation}\label{eq:lower}
\psi^l_{c,m}(x):=|x|^{-\frac c2-\frac d2}\prod_{j=1}^m\ln_j^{-1/2}(|x|)\, \text{ for all } |x|>e_m\,.
\end{equation}
Note that this function is chosen in such a way that the expression $$\int_{|x|>e_m}|x|^{\tilde c}|\psi^l_{c,m}(x)|^2\mathrm dx$$ is infinite if and only if $\tilde c\geq c$, i.e. the function $\psi^l_{c,m}$ barely fails to have finite $c$-th moment. However this function is still in the local form domain $Q^{U_{e_m}}_{loc}(H)$, where $U_R:=\{x\in\R^d:|x|>R\}$. It is even in $L^2(U_{e_m})$ provided that $c>0$. In the following lemma we construct a potential $W_{c,m}$ for which $\psi^l_{c,m}(x)$ corresponds to a zero energy eigenstate in $U_{e_m}$ to the operator $-\Delta+W_{c,m}$, i.e. $(-\Delta+W_{c,m})\psi^l_{c,m}(x)=0$ in $U_{e_m}$.

\begin{lem}\label{lemlower}
Let $m\in\N_0$, $c\in\R^+$, then $(-\Delta+W_{c,m})\psi^l_{c,m}(x)=0$ in $U_R=\{x\in\R^d:|x|>R\}$ for all large enough $R\geq e_m$, where
\begin{equation}\label{eq:Wcm}
  \begin{split}
W_{c,m}\coloneqq &\frac{d(4-d)+c^2+4c}{4|x|^2}+\frac{c+2}{2|x|^2}\sum_{k=1}^m\prod_{j=1}^k\ln_j^{-1}(|x|)+\frac{1}{4|x|^2}\left(\sum_{k=1}^m\prod_{j=1}^k\ln_j^{-1}(|x|)\right)^2\\&\quad+\frac{1}{2|x|^2}\sum_{i=1}^m\sum_{j=1}^i\prod_{k=1}^i\prod_{l=1}^j\ln_k^{-1}(|x|)\ln_l^{-1}(|x|)
 \end{split}
\end{equation}
is well defined for $|x|>e_m$.
\end{lem}

\begin{proof}
It is not hard to see that as long as $W_{c,m}(x)=\frac{\Delta\psi^l_{c,m}(x)}{\psi^l_{c,m}(x)}$, then $(-\Delta+W_{c,m})\psi^l_{c,m}(x)=0$ in $U_R$ for any $R\geq e_m$. Using notation $r=|x|$ we can write
\begin{equation}\label{eq:laplace}
\Delta f(r)=\partial_{rr} f(r)+\frac{d-1}{r}\partial_r\psi(r)\,.
\end{equation}
By a direct calculation one can check that
\begin{equation}\label{eq:lemma:non}
  \begin{split}
\partial_r\psi^l_{c,m}(x)&=-\psi^l_{c,m}(x)\left(\frac{c+d}{2r}+\frac1{2r}\sum_{k=1}^m\prod_{j=1}^k\ln_j^{-1}(r)\right)\,,\\
\partial_{rr}\psi^l_{c,m}(x)&=\psi^l_{c,m}(x)\left(\frac{c+d}{2r}+\frac1{2r}\sum_{k=1}^m\prod_{j=1}^k\ln_j^{-1}(r)\right)^2+\frac{c+d}{2r^2}\psi^l_{c,m}(x)
\\&\quad+\frac1{2r^2}\psi^l_{c,m}(x)\sum_{k=1}^m\prod_{j=1}^k\ln_j^{-1}(r)+\frac1{2r^2}\psi^l_{c,m}(x)\sum_{i=1}^m\sum_{j=1}^i\prod_{k=1}^i\prod_{l=1}^j\ln_k^{-1}(r)\ln_l^{-1}(r)
 \end{split}
\end{equation}
where one uses
$$
\partial_r\ln_1(r)=\frac1r\quad\text{and}\quad \partial_r\ln_j(r)=\frac1{\ln_{j-1}(r)}\frac1{\ln_{j-2}(r)}\ldots\frac1{\ln_{1}(r)}\frac1r\,.
$$
Now combining \eqref{eq:laplace} and \eqref{eq:lemma:non} we obtain $W_{c,m}$ in form \eqref{eq:Wcm} which completes the proof 
\end{proof}

\begin{proof}[Proof of Theorem \ref{absence}]
First we show that for any potential $V$ satisfying the assumption of Theorem \ref{absence} the function $\psi^l_{c,m}(x)$ is a subsolution. Due to the fact that $\psi^l_{c,m}(x)\in Q^{U_{e_m}}_{loc}(H)$ we can write for any positive $\varphi\in C_0^\infty(U_{e_m})$
\begin{equation}
  \begin{split}
\langle\varphi,(-\Delta+V)\psi^l_{c,m}\rangle\leq\langle\varphi,(-\Delta+W_{c,m})\psi^l_{c,m}\rangle=0
 \end{split}
\end{equation}
where we have used that $\psi^l_{c,m}(x)$ is positive and 
$$
\frac{1}{|x|^2}\left(\sum_{k=1}^m\prod_{j=1}^k\ln_j^{-1}(|x|)\right)^2+\frac{1}{2|x|^2}\sum_{i=1}^m\sum_{j=1}^i\prod_{k=1}^i\prod_{l=1}^j\ln_k^{-1}(|x|)\ln_l^{-1}(|x|)>0\text{ for }x\in U_{e_m}\,.
$$
Furthermore due to the fact that $V$ is in the local Kato class we have that the eigenfunctions of $H$ are continuous \cite{AizSim82, Sim82,Simader90}. Also the ground state eigenfunction can be chosen to be positive \cite{Far72, Goe77, ReeSim4}. 

Now assuming that $\psi$ is a positive ground state eigenfunction of $H$ we define $c_1^R:=\sup_{R\leq|x|\leq R+1}\psi$ and $c_2^R:=\sup_{R\leq|x|\leq R+1}\psi^l_{c,m}(x)$. Clearly the constant $C=\frac{c_2^R}{c_1^R}$ is finite and positive which implies that $C\psi(x)\geq \psi^l_{c,m}(x)$ for $R\leq|x|\leq R+1$. Using Theorem \ref{thm:comp-agm} we obtain
$$
\psi^l_{c,m}(x)\leq C\psi(x)
$$
for all $|x|>R$. The function $\psi^l_{c,m}(x)$ does not have finite $c$-th momentum and due to the previous inequality the function $\psi$ can not have finite $c$-th momentum because
$$
\int_{U_{e_m}}|x|^c|\psi(x)|^2\mathrm dx\geq\frac1{C^2}\int_{U_{e_m}}|x|^c|\psi^l_{c,m}(x)|^2\mathrm dx=\infty\,.
$$
\end{proof}

\section{Proof of Existence Result (Theorem \ref{existence})}
The complicated and delicate task of showing the existence of a $L^2$ zero energy ground state was already treated in \cite[Theorem 1.7]{HunJexLan22-Potential} and due to the form of the potential bound \eqref{eq:thm_existence} we can use this result. The remaining part of the existence proof follows a similar pattern as the non-existence result in Section \ref{sec:abs}. The main difference is that the ground state function $\psi$ now plays the role of the subsolution. To this end we start by taking the function
\begin{equation}\label{eq:upper}
\begin{split}
\psi^u_{c,0}(x)&:=|x|^{-\frac c2-\frac d2-\frac\varepsilon2}\, \text{ for all } |x|>e_0\,,\\
\psi^u_{c,m}(x)&:=|x|^{-\frac c2-\frac d2}\prod_{j=1}^m\ln_j^{-1/2}(|x|)\ln_m^{-\epsilon/2}(|x|)\, \text{ for all } |x|>e_m\,,m\in\N\,,
\end{split}
\end{equation}
which has barely integrable $c$-th momentum, i.e.,
$$\int_{|x|>e_m}|x|^{\tilde c}|\psi^u_{c,m}(x)|^2\mathrm dx<\infty$$
if and only if $\tilde c\leq c$. We need to construct the potential for which $\psi^u_{c,m}$ is a ground state eigenfunction. It is a straightforward calculation to find the potential $W^{c,m}$ for which the function $\psi^u_{c,m}(x)$ is a zero energy ground state in $U_{e_m}$.

\begin{lem}
Let $m\in\N_0$, $\epsilon>0$, $c\in\R^+$, then $(-\Delta+W^{c,m})\psi^u_{c,m}(x)=0$ in $U_R=\{x\in\R^d:|x|>R\}$ for all large enough $R\geq e_m$, where
\begin{equation}\label{eq:Wucm}
  \begin{split}
W^{c,0}\coloneqq &\frac{d(4-d)+c^2+4c+\epsilon(2c+4+\epsilon)}{4|x|^2}\\
W^{c,m}\coloneqq &\frac{d(4-d)+c^2+4c}{4|x|^2}+\frac{c+2}{2|x|^2}\sum_{k=1}^m\prod_{j=1}^k\ln_j^{-1}(|x|)+\frac{c\epsilon+2\epsilon}{2|x|^2}\prod_{j=1}^m\ln_j^{-1}(|x|)\\&\quad+\frac{1}{4|x|^2}\left(\sum_{k=1}^m\prod_{j=1}^k\ln_j^{-1}(|x|)+\epsilon\prod_{j=1}^m\ln_j^{-1}(|x|)\right)^2\\&\quad+\frac{1}{2|x|^2}\sum_{i=1}^m\sum_{j=1}^i\prod_{k=1}^i\prod_{l=1}^j\ln_k^{-1}(|x|)\ln_l^{-1}(|x|)+\frac{\epsilon}{2|x|^2}\prod_{l=1}^m\ln_l^{-1}(|x|)\sum_{i=1}^m\prod_{k=1}^i\ln_k^{-1}(|x|)
 \end{split}
\end{equation}
is well defined for $|x|>e_m$.
\end{lem}

\begin{proof}
We will use a large portion of the calculations done in the proof of Lemma \ref{lemlower}. With slight abuse of notation we can write for spherically symmetric functions
\begin{equation}\label{eq:delta}
\Delta(f(r)g(r))=g(r)\Delta f(r)+2\partial_r f(r)\partial_rg(r)+f(r)\Delta g(r)\,.
\end{equation}
Clearly $\psi^u_{c,m}(x)=\psi^l_{c,m}(x)\ln_m^{-\epsilon/2}(|x|)$. It is straightforward to see that
\begin{equation}\label{eq:lemma:exi}
  \begin{split}
\partial_r\ln_m^{-\epsilon/2}(r)&=-\frac{\epsilon}{2r}\ln_m^{-\epsilon/2}(r)\prod_{j=1}^m\ln_j^{-1}(r)\,,\\
\partial_{rr}\ln_m^{-\epsilon/2}(r)&=\ln_m^{-\epsilon/2}(r)\left(\frac{\epsilon}{2r}\prod_{j=1}^m\ln_j^{-1}(r)\right)^2+\frac{\epsilon}{2r^2}\ln_m^{-\epsilon/2}(r)\prod_{j=1}^m\ln_j^{-1}(r)
\\&\quad+\frac\epsilon{2r^2}\ln_m^{-\epsilon/2}(r)\prod_{j=1}^m\ln_j^{-1}(r)\sum_{i=1}^m\prod_{k=1}^i\ln_k^{-1}(r)\,.
 \end{split}
\end{equation}
Combining \eqref{eq:delta}, \eqref{eq:lemma:non}, \eqref{eq:Wcm},  and \eqref{eq:lemma:exi} we get the formula \eqref{eq:Wucm}.
\end{proof}

\begin{rem}
We note that $W^{\tilde c,m}<W^{c,m}$ for $\tilde c<c$. This also implies that $\psi^u_{\tilde c,m}$ is a supersolution for any $W^{c,m}$ as long as $\tilde c\leq c$.
\end{rem}

\begin{proof}[Proof of Theorem \ref{existence}]
Due to the assumptions on the potential $V$ \cite[Theorem 1.7]{HunJexLan22-Potential} gives that there exists a $\psi\in H^2(\R^d)$ such that $H\psi=0$. Now we need to check that such function $\psi$ has finite moments $|x|^{\tilde c}$ for  $c\in(0,c]$. We need to construct a suitable upper bounds for the function $\psi$. First we show that the function $\psi^u_{c,m}$ is a supersolution of the operator $H$ with energy $0$. This follows directly from
$$
W^{c,m}\leq\frac{d(4-d)+4c+c^2}{4|x|^2}+\frac{c+2}{2|x|^2}\sum_{j=1}^m\prod_{k=1}^j\left(\ln_k(|x|)\right)^{-1}+\frac{\varepsilon(c+2)}{2|x|^2}\prod_{k=1}^m\left(\ln_k(|x|)\right)^{-1}\leq  V(x)
$$
where the first inequality holds for sufficiently large $|x|$ and $\epsilon=\frac{\varepsilon}{2}$ and the fact that $\psi^u_{c,m}$ is positive.
Now using Theorem \ref{thm:comp-agm} we obtain $\psi(x)\leq C\psi^u_{c,m}(x)$ for a given $C\in\mathbb R^+$ for all $x\geq R$. Writing
$$
\langle\psi,|x|^{\tilde c}\psi\rangle=\int_{\mathbb R^d}|x|^{\tilde c}|\psi|^2\mathrm dx\leq R^{\tilde c}\int_{\mathbb B_R(0)}|\psi|^2\mathrm dx+\int_{\mathbb B_R^c(0)}|x|^{\tilde c}|\psi^u_{c,m}(x)|^2\mathrm dx<\infty
$$
where the fact $\int_{\mathbb B_R^c(0)}|x|^{\tilde c}|\psi^u_{c,m}(x)|^2<\infty$ for all $\tilde c\in[0,c]$ completes the proof.
\end{proof}

\section{Example Potential}
In this section we revisit the example potential presented in \cite[Appendix A]{HunJexLan22-Potential}. The potential depends on a parameter $\alpha$ and dimension $d$. For positive values of $\alpha$ the system will either have a resonance or zero energy ground state. There is also a simple relation between the value of $\alpha$ and finiteness of $c$-th momentum for the ground state. The potential which we work with can be written as

\begin{equation}\label{form:Vad}
	V_{\alpha,d}(x):=
		\frac{4\alpha^2 - (d-2)^2 }
		  {4\big(1+|x|^2\big)}
		  +\frac{1-(\alpha+d/2)^2}
		  {\big(1+|x|^2\big)^2}\,.
\end{equation}
 In the following Lemma we recall already established results within \ref{claim 1}) to \ref{claim 4}) and the new result following from the Theorems \ref{absence} and \ref{existence} in \ref{claim 5}).

\begin{lem}\label{thm: V-alpha-d}
  Let $d\in\N$, $\alpha\in\R$, and $H_{\alpha,d}= -\Delta+V_{\alpha,d}$ 
  be the self-adjoint Schr\"odinger operator with potential 
  $V_{\alpha,d}$ given above.  
  Then 
  \begin{SL}
  	\item\label{claim 1} $\sigma(H_{\alpha,d})=\sigma_{ess}(H_{\alpha,d})=[0,\infty)$. 
  	\item\label{claim 2} For all $\alpha\ge 0$ the potential $V_{\alpha,d}$ is critical, that is, the Schr\"odinger operator 
  			$H_{\alpha,d}$ has a virtual level. 
  	\item\label{claim 3} Zero is not an eigenvalue  of 
  		  $H_{\alpha,d}$ when $ 0\le \alpha\le 1$.  
  		  For  $\alpha>1$ zero is an eigenvalue. 
  		  The zero energy resonance for $0\le \alpha\le 1$, 
  		  respectively ground state for $\alpha>1$, is given by  
  		$  	\psi_{\alpha,d}(x)=(1+|x|^2)^{(2-d)/4 -\alpha/2}	$.
	\item\label{claim 4} For $\alpha<0$, the potential 
			$V_{\alpha,d}$ is subcritical, and zero is neither an eigenvalue nor a resonance. 
	\item\label{claim 5} For $\alpha>1$, the moment $\langle\psi_{\alpha,d}(x),|x|^c\psi_{\alpha,d}(x)\rangle<\infty$ for all $c<2(\alpha-1)$. 
  \end{SL}
\end{lem}

\begin{proof}
The proof of claims \ref{claim 1}) to \ref{claim 4}) can be found in \cite[Appendix A]{HunJexLan22-Potential}. The proof of claim \ref{claim 5}) is a direct consequence of Theorems \ref{absence} and \ref{existence}. This can be seen by solving the equation $4\alpha^2-(d-2)^2=(4-d)d+c^2+4c$ for $c$. The equation follows by comparing the value of the prefactor in front of $\frac1{4|x|^2}$ in Theorems \ref{absence} and \ref{existence} with the prefactor  in front of $\frac1{4(|x|+1)^2}$ in \eqref{form:Vad}.
\end{proof}



\section*{Acknowledgments}
The author would like to thank Markus Lange for useful suggestions and proofreading the manuscript. The author received financial support from the Ministry of Education, Youth and Sport of the Czech Republic under the Grants No. RVO 14000.

\bibliographystyle{acm}
 \bibliography{references}

\end{document}